\documentclass[pra,twocolumn,showpacs,superscriptaddress]{revtex4}
\usepackage{amsmath,amsfonts,amssymb,amsthm,algorithm,algorithmic,blackboard}

\usepackage{bm}
  \usepackage[pdftex]{graphicx}
  \graphicspath{{figures/pdf/}}
  \DeclareGraphicsExtensions{.pdf}
\usepackage{subfig}
\newtheorem{theorem}{Theorem}
\newtheorem{corollary}{Corollary}

\newtheorem{example}{Example}
\def\H{\mathcal{H}}
\def\L{\mathcal{L}}

\def\uu{\mathfrak{u}}
\def\su{\mathfrak{su}}
\def\SU{\mathfrak{SU}}
\def\UU{\mathfrak{U}}
\def\LL{\mathfrak{L}}
\def\vec#1{\mathbf{#1}}
\def\ket#1{| #1 \rangle}
\def\bra#1{\langle #1 |}
\def\op#1{\hat{#1}}

\def\norm#1{\| #1 \|}

\def\Tr{\mathop{\rm Tr}}
\def\Had{\mathop{\rm Had}}
\def\CNOT{\mathop{\rm CNOT}}
\def\sx{\vec{\sigma}_x}

\def\sz{\vec{\sigma}_z}

\begin{document}
\title{Global controllability with a single local actuator}
\author{S.~G.~Schirmer}\email{sgs29@cam.ac.uk}
\author{I.~C.~H.~Pullen}
\author{P.~J.~Pemberton-Ross} 
\address{Dept of Applied Maths and Theoretical Physics,
University of Cambridge, Wilberforce Rd, Cambridge, CB3 0WA, UK.}
\date{\today}
\begin{abstract}
We show that we can achieve global density-operator controllability for most $N$-dimensional bilinear Hamiltonian control systems with general fixed couplings using a \emph{single, locally-acting} actuator that modulates one energy-level transition. Controllability depends upon the position of the actuator and relies on the absence of either decompositions into non-interacting subgroups or symmetries restricting the dynamics to a subgroup of $\SU(N)$.  These results are applied to spin-chain systems and used to explicitly construct control sequences for a single binary-valued switch actuator.  
\end{abstract}
\pacs{03.67.-a,75.10.Pq,78.67.Lt,02.20.Yy,07.05.Dz}
\maketitle

\section{Introduction}

The ability to control the dynamics of quantum systems is a long
established objective in areas as diverse as molecular chemistry and
quantum computing among others.  Control in practice comprises various
related tasks such as transforming a system from a given initial state
to a desired target state, implementing a desired unitary operator, or
optimising the expectation value of a selected observable.  The manner
in which control is effected depends on the system but a common approach
for quantum systems is the application of external electromagnetic
fields.  In the diabatic control regime these drive transitions between
different states of the system, and control can be achieved by adjusting
the amplitude and phase of the field(s) as a function of time in a way
that maximizes constructive interference of various excitation pathways
that lead to a desired outcome, while maximizing destructive interference
for all others.

Although the ultimate goal of control is usually to find a control field
that steers the system in the manner required to achieve the objective,
the question of what tasks can be accomplished for a given system with a
given set of actuators, is of fundamental interest.  A key concept in
this regard is that of controllability.  A substantial number of papers 
have been devoted to studying this issue for both classical and quantum
systems, establishing various notions of controllability and general 
algebraic criteria for them, and showing that particular types of systems 
are controllable~\cite{72JS,JMP24p2608,PRA51p0960,qph0106128,PRA63n063410,
JPA35p4125,qph0302121,JOB7pS293}.  On the latter front, it has been shown, 
for instance, that any system with distinct transition frequencies and a 
connected transition graph is controllable~\cite{CP267p001,JMP43p2051}.  
For an $n$-level system this requires at least $n-1$ transitions with
non-zero probabilities.  It has also been shown that these requirements 
can be relaxed in many cases~\cite{JPA34p1679}, and more recently indirect 
controllability has been studied~\cite{PRA75n052317}.

One remaining area of interest is global controllability with a small
number of local actuators.  A motivation for this type of scenario could
be a chain or array of quantum dots with control electrodes to locally
manipulate the dynamics of one or a few quantum dots.  In the ideal
case, one might consider separate control electrodes for each quantum
dot, as well as separate electrodes to modulate all the interactions
between pairs of adjacent quantum dots, as proposed by Kane
in~\cite{NAT393p133} and many other quantum computing architectures
since.  Leaving aside the often considerable challenge of finding
optimal control schemes and fighting decoherence, with sufficiently many
local actuators almost any (Hamiltonian) quantum system is controllable,
at least in principle.  However, in many cases it is impractical or even
impossible to have a large number of individual local actuators such as
control electrodes.  Rather, one would like to make do with as few local
actuators as possible to simplify the engineering design and reduce
deleterious effects such as decoherence and crosstalk, for example.

Motivated by this problem we investigate the question of controllability
of a finite-dimensional model system with the smallest number of simple
actuators whose effect is strictly confined a local perturbation of the
Hamiltonian.  We also note here that by local we mean localized in space, 
affecting a single transition, for instance, not simultaneous local 
operations on many individual elements such as qubits as is common in 
global control schemes.  We effectively show that in most cases a single 
local actuator is sufficient to ensure controllability of the system as 
a whole, provided the latter is not decomposable into non-interacting 
parts, and does not exhibit dynamical symmetries that its evolution to a 
subgroup of the unitary group ($\UU(N)$ or $\SU(N)$).  Many systems with 
fixed interactions connecting its parts such as chains of quantum dots 
etc with fixed non-zero couplings between adjacent dots satisfy this 
connectedness requirement, and the disorder present in most realistic 
systems is likely to ensure that there are no special dynamical symmetries 
to worry about in most cases.  For these systems our controllability
analysis suggests that the entire system can be controlled by modulating
a single transition with a local actuator.  Although the explicit
controllability proofs given apply to specific model systems, the same
arguments are applicable to many other model systems, suggesting that a
large class of systems with fixed couplings may be controllable using a
very small number of fixed local actuators.  We conclude with an explicit 
example of constructive control with a single binary switch actuator.

\section{Definitions and basic results on controllability}

We restrict ourselves here to control problems that can be classified as
open-loop Hamiltonian engineering problems and systems subject to
Hamiltonian dynamics.  Open-loop control engineering means that we aim
to design control fields relying only on (presumed) knowledge of the
initial state of the system and the dynamic laws governing its evolution
in the presence of the control fields, without any feedback from
measurements.  We furthermore assume the state space of the system is
a finite-dimensional Hilbert space, $\H\simeq \CC^N$.  The state of the
system in this case can be represented by a density operator $\op{\rho}$, 
i.e., a positive unit-trace operator acting on $\H$, and its evolution
is governed by the quantum Liouville equation
\begin{equation}
\label{quantum Liouville eqn}
i \hbar \frac{d}{dt} \hat{\rho}(t)
 = \left[\hat{H}[\mathbf{f}(t)],\hat{\rho}(t) \right]+
i \hbar \L_D [\hat{\rho}(t)], 
\end{equation}
where $[A,B]=AB-BA$ is the usual matrix commutator and $\L_D=0$ for a
Hamiltonian control system.  The operator $\hat{H}[\mathbf{f}(t)]$ is 
the total Hamiltonian of the system subject to the control fields 
$\vec{f}(t)$.  For control-linear systems we have the perturbative
expansion
\begin{equation}
\label{Hamiltonian}
  \hat{H}[\mathbf{f}(t)]=\hat{H}_0+ \sum_{m=1}^Mf_m(t)\hat{H}_m, 
\end{equation}
where $\hat{H}_0$ is the internal Hamiltonian of the system and
$\hat{H}_m$, $m>0$, are the interaction terms.  

Hamiltonian dynamics constrains the evolution of density operators 
$\rho(t)$ to isospectral flows 
\begin{equation}
\hat{\rho}(t)=\hat{U}(t,t_0)\hat{\rho}_0\hat{U}(t,t_0)^{\dagger},
\end{equation}
since the evolution operator $\hat{U}(t,t_0)$ must satisfy the related 
Schrodinger equation
\begin{equation}
i\hbar \frac{d}{dt}\hat{U}(t,t_0)=\hat{H}[\mathbf{f}(t)]\hat{U}(t,t_0)
\end{equation}
and is hence restricted to the unitary group $\UU(N)$.  Due to this 
fundamental restriction it is clear that the maximum degree of state
control we can achieve for this system is the ability to interconvert
density operators with the same spectrum, which is achieved if we can
implement any unitary operator in the special unitary group $\SU(N)$
of unitary operators with determinant $1$ as abelian factors do not
affect the isospectral flow.  It is also not difficult to show that 
any proper subgroup of $\SU(N)$ is not sufficient to interconvert any
two generic density operators with the same spectrum.  

To properly define the notion of controllability we need some concepts
from Lie group / algebra theory.  A \emph{Lie algebra} is a vector 
space over a field endowed with a bilinear composition $[x,y]$ that 
satisfies the Jacobi identity 
\begin{equation*}
 [[x,y],z]+[[y,z],z]+[[z,x],y]=0.
\end{equation*}
It is easy to see that the anti-Hermitian matrices $iH_0$ and $iH_1$
generate a Lie algebra $\LL$ which must be a subalgebra of the Lie 
algebra of skew-hermitian matrices $\uu(N)$, and if $iH_0$ and $iH_1$ 
have zero trace, $\LL$ will be a subalgebra of the trace-zero,
anti-Hermitian matrices $\su(N)$, which can be regarded as the tangent
space to the Lie group $\SU(N)$ at the identity via the exponential
map $x\in\su(N)\mapsto\exp(x)\in\SU(N)$.  Therefore, we can argue that 
if the $i\hat{H}_0$ and $i\hat{H}_1$---or their trace-zero counterparts 
$\tilde{H}_m=\hat{H}_m-N^{-1}\Tr(\hat{H}_m)I_N$---generate the entire 
Lie algebra $\su(N)$ then we can in principle dynamically generate any 
matrix $\op{U}\in \SU(N)$.  Hence, a system is said to be density matrix 
controllable or simply \emph{controllable} if the Lie algebra generated 
by $i\tilde{H}_0$ and $i\tilde{H}_1$ is $\su(N)$.  These Lie algebraic
criteria are useful as they are easy to check by quite straightforward 
calculations.  In principle, these can be done numerically for a given 
set of Hamiltonians but for higher dimensional systems the calculations 
can be time-consuming and the accuracy very limited.  It is therefore
desirable to have more explicit criteria that guarantee controllability
for certain classes of systems, and several such results exist.  

For example, consider a simple finite-dimensional system ($\dim \H=N$) 
with a control-linear Hamiltonian of the form $\op{H}_0+f(t)\op{H}_1$.
Choose a basis such that $H_0$ is diagonal with energy levels $E_n$, 
$n=1,\ldots,N$, and transition frequencies $\omega_{mn}=E_{n}-E_m$.  If 
$\op{H}_0$ is regular, i.e., has non-degenerate eigenvalues then we can
associate each 1-dimensional eigenspace with the vertex of a graph and
interpret the non-zero elements in the matrix representation of the
interaction Hamiltonian $\op{H}_1$ (with respect to the eigenbasis of 
$H_0$) as edges of a transition graph.  In this case a sufficient 
condition for controllability is that $H_0$ be strongly regular, i.e.,
have distinct transition frequencies $\omega_{mn}\neq\omega_{m'n'}$ 
unless $(m,n)=(m',n')$, and the transition graph as defined above be 
connected~\cite{CP267p001}.  The conditions of uniqueness of the transition
frequencies can be slightly relaxed in that we only need to consider 
the transition frequencies of those transitions that occur with non-zero 
probability.  This is a useful result as it is very easy to check, 
although it is important to remember that it provides only a sufficient, 
not a necessary condition, and indeed many systems that do not satisfy 
these conditions are controllable. For example, given a system with a
Hamiltonian of the form $\op{H}[f(t)]=\op{H}_0+f(t)\op{H}_1$, where
\begin{subequations}
\label{eq:H}
\begin{align}
  H_0 &= \sum_{n=1}^N E_n \ket{n}\bra{n}, \\
  H_1 &= \sum_{n=1}^{N-1} d_n [\ket{n}\bra{n+1}+\ket{n+1}\bra{n}],
\end{align}
\end{subequations}
the graph connectivity result allows us to conclude that the system is
controllable provided the energy levels of the system are such that the
frequencies of all transitions between adjacent states are distinct and
$d_n\neq0$ for $n=1,2,\ldots,N-1$.  In principle this controllability
result can be explained in terms of frequency-selective control.  If all
the possible transitions have different frequencies then we can imagine
a field resonant with a particular transition frequency as selectively
driving only the resonant transition.  Thus, we can implement all the
generators $\op{x}_{n,n+1}$ and $\op{y}_{n,n+1}$ of the Lie algebra and
therefore the entire Lie algebra, implying controllability.  However, it
has been shown that many systems that do not satisfy this condition such
as the truncated harmonic oscillator with $d_n=\sqrt{n}$ are 
controllable despite \emph{all} allowed transitions having the same
frequency ~\cite{JPA34p1679}.  What permits controllability in this
case, despite the lack of any frequency selectivity, are the differences
in the transition strengths $d_n$.  If both the frequencies and
strengths are the same for all transitions, or even if they satisfy
certain not necessarily obvious symmetries, then controllability is
indeed lost~\cite{JPA35p2327}.

Furthermore, both the graph connectivity and the explicit Lie algebraic
results mentioned above apply only to systems where the field drives all
the possible transitions, i.e., is global as shown in Fig.~\ref{fig:1a}.
Given a particular control field that drives only a single transition,
or a subset of all possible transitions (see Fig.~\ref{fig:1b}), we
cannot draw any conclusions even if the transition graph of the system
appears connected and the transition frequencies are distinct.  Of
course, if we have many local control fields, each selectively driving a
single transition, as shown in Fig~\ref{fig:1c} then it is again obvious
that the system is controllable.  It is not obvious, however, under
which conditions the ability to control a single transition of a large
connected system as in Fig~\ref{fig:1b} is sufficient for
controllabilty.

\begin{figure}
\subfloat[Global control field simultaneously driving all transitions]{\scalebox{0.4}{\includegraphics{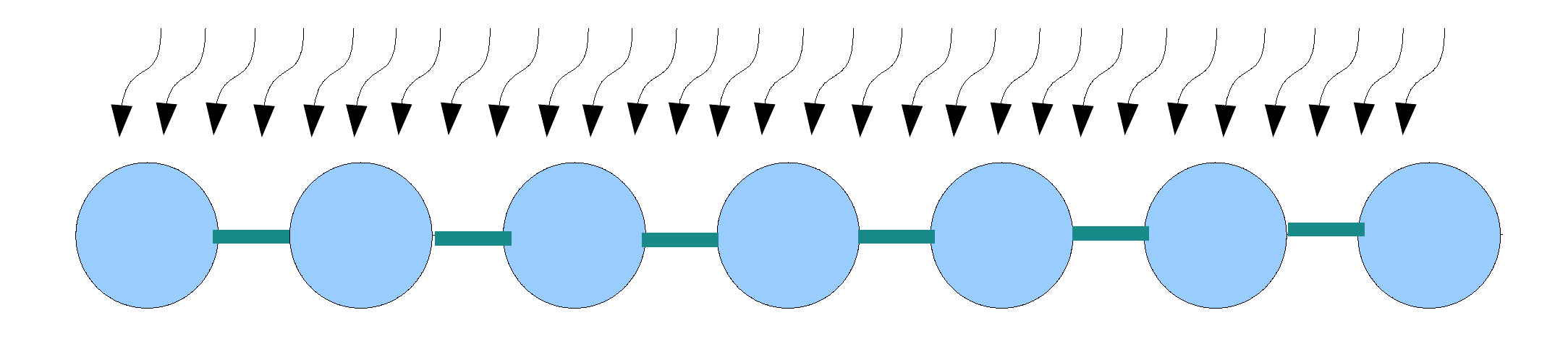}} \label{fig:1a}}\\
\subfloat[Single local actuator driving a single transition]{\scalebox{0.4}{\includegraphics{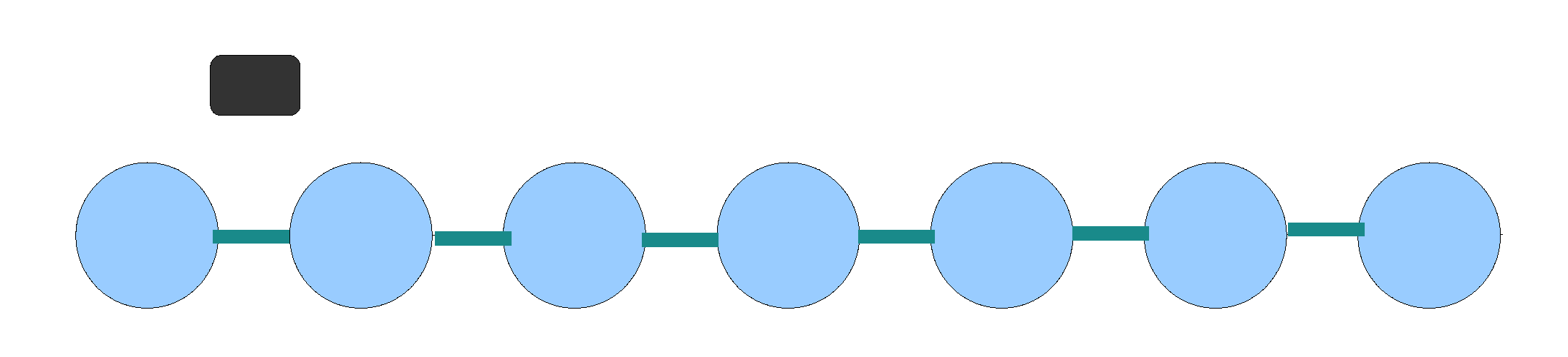} \label{fig:1b}}}\\
\subfloat[Many local actuators driving individual transitions]{\scalebox{0.4}{\includegraphics{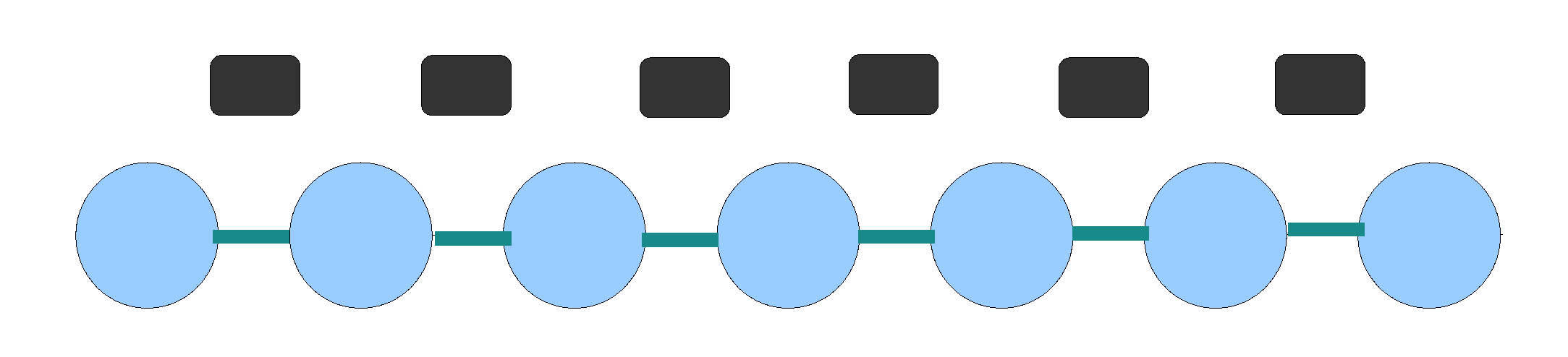} \label{fig:1c}}}\\
\caption{Schematic of $N$-state system with coupling between adjacent
 states with (a) global control field, (b) a single local actuator and
(c) many local actuators.}
\end{figure}

\section{Controllability for Single Local Actuator}

In this section we consider a model system with $N$ distinct states that 
are permanently coupled in some form, such as an array of quantum dots.  
For our model we first assume coupling of nearest neighbour type, which 
leads to a drift Hamiltonian of the form
\begin{equation}
 \label{eq:A0}
  A_0=H_0+H_1
\end{equation}
with $H_0$ and $H_1$ as in Eq.~(\ref{eq:H}), and a single local actuator 
modulating the coupling between states $\ket{r}$ and $\ket{r+1}$,
\begin{equation}
 \label{eq:Ar}
  A_r=\ket{r}\bra{r+1}+\ket{r+1}\bra{r}.
\end{equation}  
Thus, for a single local actuator positioned between $r$ and $r+1$ 
we have the total Hamiltonian
\begin{equation}
\label{eq:Htotal}
  H[f(t)] = A_0 + f(t) A_r.
\end{equation} 

\begin{example}
The Hamiltonian of the first excitation subspace of a spin chain of
length $N$ with nearest neighbour coupling of isotropic Heisenberg
form given by the coupling constants $d_n>0$ for $n=1,\ldots,N-1$ is
\begin{equation}
 \label{eq:Hspin}
  A_0 = \sum_{n=1}^N E_n \ket{n}\bra{n}
     +\sum_{n=1}^{N-1} d_n [\ket{n}\bra{n+1}+\ket{n+1}\bra{n}],
\end{equation}
where the energy levels are explicitly
\begin{equation} 
 \label{eq:Espin}
 E_n = \frac{1}{2}\sum_{\ell\neq n-1,n} d_\ell -\frac{1}{2}(d_{n-1}+d_n)
\end{equation}
and we set $d_0=d_N=0$.  Assuming we have a local actuator that 
allows us to modulate the coupling between spins $r$ and $r+1$,
the total Hamiltonian is of the form Eq.~(\ref{eq:Htotal}) with 
$A_r$ of the form Eq.~(\ref{eq:Ar}).  The first excitation subspace
Hamiltonian for spin chains with dipole-dipole interactions is also
of form~(\ref{eq:Hspin}) but with different energy levels.  Similar 
results hold for any spin chain decomposable into excitation subspaces.
\end{example}

Again, it is quite obvious that $N-1$ independent local actuators 
of this type, controlling the coupling between spins $n$ and $n+1$ 
in the chain, will suffice for the system to be controllable, but 
in fact, a single such actuator suffices in most cases.

\begin{theorem}
\label{thm:1}
A quantum system with Hamiltonian $H[f(t)]=A_0+f(t)A_r$ with $A_0$ 
and $A_r$ as above is controllable if $\omega_r \neq 0$, $d_n \neq 0$ 
and $d_{r+1}^2 \neq d_{r-1}^2$.
\end{theorem}

\begin{proof}
We show that the trace-zero anti-Hermitian matrices $iV_0$ and $iV_r$
defined by
\begin{align*}
 V_0 = A_0 - \frac{\Tr(A_0)}{N} I_N, \quad
 V_1 = A_r - \frac{\Tr(A_r)}{N} I_N
\end{align*}
generate the Lie algebra $\LL=\su(N)$.  To this end it suffices to 
show that the Lie algebra $\LL$ contains the $2(n-1)$ generators 
$x_n \equiv x_{n,n+1}$ and $y_n \equiv y_{n,n+1}$ of $\su(N)$, where 
the basis elements of $\su(N)$ are defined as usual,
\begin{align*}
 x_{mn} &= \ket{n}\bra{m}-\ket{m}\bra{n}, \\
 y_{mn} &= i(\ket{n}\bra{m}+\ket{m}\bra{n}), \\
 h_{n}  &= \ket{n}\bra{n}-\ket{n+1}\bra{n+1},
\end{align*}
for $1\le m<n\le N$.  Let $V_0^{(0)}=i(V_0-d_r V_1)$.  We have 
$iV_1=y_r\in\LL$ and
\begin{equation*}
\begin{split}
 X_0 &\equiv [y_r, V_0^{(0)}]=d_{r-1}x_{r-1,r+1}-d_{r+1}x_{r,r+2}-\omega_r x_r \\
 Y_0 &\equiv [X_0, y_r]=d_{r-1}y_{r-1}+d_{r+1}y_{r+1}-2\omega_r h_r\\
 X_0'&\equiv [Y_0, y_r]=-d_{r-1}x_{r-1,r+1}+d_{r+1}x_{r,r+2}+4\omega_r x_r\\
 Y_0'&\equiv [X_0', y_r]=-d_{r-1}y_{r-1}-d_{r+1}y_{r+1}+8\omega_r h_r
\end{split}
\end{equation*}
yields $x_r=(3 \omega_r)^{-1} (X_0+X_0')\in\LL$ 
and $h_r=2^{-1}[x_r,y_r] \in\LL$ as $\omega_r\neq0$.  Next setting
\begin{align*}
 Y_1 &\equiv 3^{-1}(4Y_0+Y_0')=d_{r-1}y_{r-1}+d_{r+1}y_{r+1}\\
 X_1 &\equiv [[x_r,Y_1],y_r]=d_{r-1}x_{r-1}+d_{r+1}x_{r+1}\\
 Z_1 &\equiv 2^{-1}[X_1,Y_1]=d^2_{r-1}h_{r-1}+d^2_{r+1}h_{r+1}\\
 Y_1'&\equiv 2^{-1}[Z_1,X_1]=d^3_{r-1}y_{r-1}+d^3_{r+1}y_{r+1}\\
 X_1'&\equiv 2^{-1}[Y_1,Z_1]=d^3_{r-1}x_{r-1}+d^3_{r+1}x_{r+1},
\end{align*}
and $c_1=d_{r-1}^2-d_{r+1}^2$ leads to
\begin{align*}
 Y_1'-d_{r+1}^2Y_1&= d_{r-1} c_1 y_{r-1}, \\
 X_1'-d_{r+1}^2X_1&= d_{r-1} c_1 x_{r-1}, \\
 Y_1'-d_{r-1}^2Y_1&= -d_{r+1} c_1 y_{r+1}, \\
 X_1'-d_{r-1}^2X_1&= -d_{r+1} c_1 x_{r+1}.
\end{align*}
Since $d_{r\pm 1}\neq 0$, $c_1\neq0$ by hypothesis, we have 
$y_{r\pm 1}$, $x_{r\pm 1}$, and $h_{r\pm1}=2^{-1}[x_{r\pm1},y_{r\pm1}]$
in $\LL$.  Next note that
\begin{equation*}
  V_0^{(1)}\equiv V_0^{(0)}-Y_1 = iH_0 +\sum_{n \in I^{(1)}} d_n y_n
\end{equation*}
where $I^{(1)}$ is the index set $\{1,\ldots,N-1\}$ minus the subset
$\{r-1,r,r+1\}$ and we have 
\begin{align*}
  Y_2'     &\equiv [[Z_1,V_0^{(1)}],Z_1 ]
                    = d_{r-2}d_{r-1}^4 y_{r-2} + d_{r+1}^4 d_{r+2}y_{r+2}\\
 V_0^{(2)} &\equiv V_0^{(1)}-d_{r-1}^{-4}Y_2'
                    = iH_0 + \sum_{n \in I^{(2)}} d_n y_n + c_{r+2}y_{r+2},
\end{align*} 
with $I^{(2)}$ the index set $I^{(1)}$ minus $\{r-2,r+2\}$ and 
$c_{r+2}=d_{r+2}(1-d_{r+1}^4/d_{r-1}^4)$.  Hence
\begin{align*}
X_2 &\equiv [Z_1,V_0^{(2)}] = d_{r+1}^2 c_{r+2}x_{r+2}, \\
Y_2 &\equiv [X_2,Z_1]       = d_{r+1}^4 c_{r+2}y_{r+2}
\end{align*}
shows $x_{r+2}$, $y_{r+2}$ and $h_{r+2}=2^{-1}[x^{r+2},y_{r+2}]$ in $\LL$.  
Setting $V_0^{(3)}=V_0^{(2)}-d_{r+2}y_{r+2}$ now shows that
\begin{align*}
x_{r+3}&=d_{r+3}^{-1}[h_{r+2},V_0^{(3)}] \in \LL, \\
y_{r+3}&=[x_{r+3},h_{r+2}]\in \LL, \\
h_{r+3}&=2^{-1}[x_{r+3},y_{r+3}] \in \LL.
\end{align*}
Repeating this procedure with $V_0^{(k+1)}=V_0^{(k)}-d_{r+k}y_{r+k}$
we obtain
\begin{align*}
x_{r+k+1}&=d_{r+k+1}^{-1}[h_{r+k},V_0^{(k+1)}]\in \LL\\
y_{r+k+1}&=[x_{r+k+1},h_{r+k}]\in \LL \\
h_{r+k+1}&=2^{-1}[x_{r+k+1},y_{r+k+1}] \in \LL
\end{align*}
for $3 \leq k \leq N-r-2$.  
To show that the elements $x_{r-k}$, $y_{r-k}$ for $2 \leq k \leq r-1$ 
are in $\LL$, we note that
\begin{align*}
y_{r-2}&=d_{r-2}^{-1}d_{r-1}^{-4}(Y_2'-d_{r+1}^4d_{r+2}y_{r+2})\in\LL \\
x_{r-2}&=[h_{r-1},y_{r-2}]\in\LL \\
h_{r-2}&=2^{-1}[x_{r-2},y_{r-2}]\in\LL
\end{align*}
and setting $W_0^{(2)}=V_0^{(N-r-1)}$ and
$W_0^{(k+1)}=W_0^{(k)}- d_{r-k-1}y_{r-k-1}$ shows
\begin{align*}
x_{r-k-1}&= d_{r-k-1}^{-1}[h_{r-k},W_0^{(k)}]\in\LL\\
y_{r-k-1}&= [x_{r-k-1},h_{r-k}]\in\LL\\
h_{r-k-1}&= 2^{-1}[x_{r-k-1},y_{r-k-1}]\in\LL
\end{align*}
for $2\leq k \leq r-2$, as desired.
\end{proof}

For a Heisenberg spin chain $d_n>0$ for $n=1,\ldots, N-1$ and 
Eq.~(\ref{eq:Espin}) show that $\omega_n=d_{n-1}-d_{n+1}$.  Thus 
$\omega_r\neq 0$ is equivalent to $d_{r+1}\neq d_{r-1}$ and we 
have the following 

\begin{corollary}
The first excitation subspace of a Heisenberg spin chain of length 
$N$ with coupling constants $d_n$ is controllable with single local 
actuator between spins $r$ and $r+1$ if $d_{r+1}\neq d_{r-1}$.
\end{corollary}

A Heisenberg spin chain with non-uniform couplings almost certainly 
satisfies $d_{r-1}^2\neq d_{r+1}^2$ for any $r$ between $1$ and $N-1$.
A chain with uniform coupling $d_n=d$, $n=1,\ldots,N-1$, $d_0=d_N=0$,
satisfies this condition only if the actuator is placed near the end
of the chain, i.e., $r=1$ or $r=N-1$.  However, we can generalize the 
previous theorem.

\begin{theorem}
\label{thm:2}
A quantum system with Hamiltonian $H[f(t)]=A_0+f(t)A_r$ with $A_0$ 
and $A_r$ as above is controllable if $\omega_r\neq 0$, $d_n\neq 0$ 
and $d_{r-k-1}^2\neq d_{r+k+1}^2$ for some $k\in\NN_0$.
\end{theorem}

\begin{proof}
For $k=0$, i.e., if $d_{r-1}^2 \neq d_{r+1}^2$, the result follows from 
Thm~\ref{thm:1}.   If $d_{r-1}^2=d_{r+1}^2$ we begin as in the proof of 
Thm~\ref{thm:1} to conclude that $y_r\in\LL$, 
$x_r=(3 \omega_r)^{-1}(X_0+X_0')\in\LL$ and $h_r=2^{-1}[x_r,y_r]\in\LL$, 
and set
\begin{align*}
 V_0^{(0)} &\equiv iV_0-d_ry_r\\
 Y_1^{(0)} &\equiv 3^{-1}(4Y_0+Y_0')=d_{r-1}y_{r-1}+d_{r+1} y_{r+1}\\
 X_1^{(0)} &\equiv [[x_r,Y_1^{(0)}],y_r]=d_{r-1}x_{r-1}+d_{r+1} x_{r+1}\\
 Z_1^{(0)} &\equiv 2^{-1}[X_1^{(0)},Y_1^{(0)}]=d_{r-1}^2h_{r-1} +d_{r+1}^2h_{r+1}\\
 V_0^{(1)} &\equiv V_0^{(0}-Y_1^{(0)}=iH_0-\sum_{n \in I^{(1)}}d_n y_n
\end{align*}
where $I^{(1)}$ is the index set $\{1, \dots, N-1\}$ minus the subset 
$\{r-1,r,r+1\}$.

Setting $d_{r+j}^2= d_{r-j}^2$ for $j=1 \dots k-1$ and observing that we 
cannot separate the $r+1$ to $r+k$ and $r-1$ to $r-k$ terms, respectively, 
at this stage we continue along similar lines by iterating the following 
set of recurrence relations  for $j=1, \dots, k-1$
\begin{align*}
Z_j^{(1)}&\equiv d_{r-j}^{-2} Z_1^{(0)} \\
         &= h_{r-j}+h_{r+j} \textrm{ as } d_{r+j}^2= d_{r-j}^2 \neq 0\\
X_j^{(1)}&\equiv [Y_j^{(0)},V_0^{(j)}] \\
         &= d_{r-j} d_{r-j-1} x_{r-j-1,r-j+1} -d_{r-j}\omega_{r-j}h_{r-j}\\
         &\quad -d_{r+j}d_{r+j+1}x_{r+j,r+j+2}-d_{r+j}\omega_{r+j}h_{r+j}\\
Y_j^{(1)}&\equiv [X_j^{(1)},Y_j^{(0)}] \\
         &= d_{r-j}^2 d_{r-j-1} y_{r-j-1}   -2d_{r-j}^2\omega_{r-j}h_{r-j}\\
         &\quad +d_{r+j}^2d_{r+j+1}y_{r+j+1}-2d_{r+j}^2\omega_{r+j}h_{r+j}\\
Y_j^{(2)}&\equiv d_{r-j}^{-2} Y_j^{(1)} \\
         &= d_{r-j-1}y_{r-j-1}-2\omega_{r-j}h_{r-j}\\
         &\quad + d_{r+j+1}y_{r+j+1} -2\omega_{r+j}h_{r+j}\\
X_{j+1}^{(0)}&\equiv [Z_j^{(1)},Y_j^{(2)}] \\
             &= d_{r-j-1} x_{r-j-1} + d_{r+j+1}x_{r+j+1}\\
Y_{j+1}^{(0)}&\equiv [X_{j+1}^{(0)},Z_j^{(1)}]\\ 
             &= d_{r-j-1}y_{r-j-1} + d_{r+j+1} y_{r+j+1}\\
Z_{j+1}^{(0)}&\equiv 2^{-1}[X_{j+1}^{(0)},Y_{j+1}^{(0)}]\\ 
             &= d_{r-j-1}^2h_{r-j-1}+d_{r+j+1}^2 h_{r+j+1}\\
V_0^{(j+1)}  &\equiv V_0^{(j)}-Y_{j+1}^{(0)}\\
             & =iH_0-\sum_{n \in I^{(j+1)}}d_ny_n
\end{align*}
where $I^{(j+1)}$ is the index set $I^{(j)}$ with the subset 
$\{r-j-1,r+j+1\}$ removed.  Since $d_{r-k-1}^2 \neq d_{r+k+1}^2$ and
\begin{align*}
 X_{k}^{(0)}&= d_{r-k}x_{r-k}  +d_{r+k} x_{r+k}\\
 Y_{k}^{(0)}&= d_{r-k}y_{r-k}  +d_{r+k} y_{r+k}\\
 Z_{k}^{(0)}&= d_{r-k}^2h_{r-k}+d_{r+k}^2 h_{r+k}\\
 V_0^{(k)}  &= iH_0-\sum_{n\in I^{(k)}}d_ny_n
\end{align*}
where $I^{(k)}$ is the index set $\{1, \dots, r-k,r+k, \dots N \}$.

To complete the proof by showing that $y_{r \pm (k+1)}$, $x_{r \pm (k+1)}$
and $h_{r \pm (k+1)}$ are in $\LL$, we calculate the commutators 
\begin{align*}
X_k^{(1)} &\equiv [Y_k^{(0)},V_0^{(k)}] \\
          &= d_{r-k}d_{r-k-1}x_{r-k-1,r-k+1}    -d_{r-k}\omega_{r-k} h_{r-k}\\
          &\quad -d_{r+k}d_{r+k+1}x_{r+k,r+k+2} -d_{r+k}\omega_{r+k} h_{r+k}\\
Y_k^{(1)} &\equiv [X_k^{(1)},Y_k^{(0)}] \\
          &= d_{r-k}^2d_{r-k-1}y_{r-k-1}        -2d_{r-k}^2\omega_{r-k}h_{r-k}\\
          &\quad +d_{r+k}^2d_{r+k+1}y_{r+k+1}   -2d_{r+k}^2\omega_{r+k}h_{r+k}\\
Z_k^{(1)} &\equiv d_{r-k}^{-2} Z_k^{(0)} \\
          &= h_{r-k}+h_{r+k} \textrm{ as } d_{r-k}^2=d_{r+k}^2
\end{align*}
\begin{align*}
X_k^{(2)} &\equiv [Z_k^{(1)},Y_k^{(1)}] \\
          &= d_{r-k}^{-2}d_{r-k-1}x_{r-k-1}+d_{r+k}^{-2}d_{r+k+1}x_{r+k+1}\\
X_k^{(3)} &\equiv d_{r-k}^{-2} X_k^{(2)}\\
          &= d_{r-k-1}x_{r-k-1}+d_{r+k+1}x_{r+k+1}\\ 
Y_k^{(2)} &\equiv [X_k^{(3)},Z_k^{(1)}] \\
          &= d_{r-k-1}y_{r-k-1}+d_{r+k+1}y_{r+k+1}\\
Z_k^{(2)} &\equiv 2^{-1}[X_k^{(3)},Y_k^{(2)}] \\
          &= d_{r-k-1}^2h_{r-k-1}+d_{r+k+1}^2h_{r+k+1}\\
Y_k^{(3)} &\equiv 2^{-1}[Z_k^{(2)},X_k^{(3)}]\\
          &= d_{r-k-1}^3y_{r-k-1}+d_{r+k+1}^3y_{r+k+1}
\end{align*}
which gives
\begin{align*}
y_{r \pm (k+1)} &= \frac{Y_k^{(3)}-d_{r \mp (k+1)}Y_k^{(2)}}
                   {d_{r \pm (k+1)}(d_{r \pm (k+1)}^2 - d_{r \mp (k+1)}^2)}\\
x_{r \pm (k+1)} &\equiv [y_{r \pm (k+1)},Z_k^{(2)}]/(2d_{r \pm (k+1)}^2)\\
h_{r \pm (k+1)} &\equiv [x_{r \pm (k+1)},y_{r \pm (k+1)}]/2 
\end{align*}
showing that these generators are in $\LL$.  To show that the generators
$x_{r \pm j}$, $y_{r \pm j}$, and $h_{r \pm j}$ are in $\LL$ for 
$j = 1,\ldots, k$ we set
\begin{align*}
V_0^{(k+1)} &= V_0^{(k)}-d_{r-k-1}y_{r-k-1}-d_{r+k+1}y_{r+k+1}\\
            &= iH_0-\sum_{n \in I^{(k+1)}}d_n y_n,
\end{align*}
where $I^{(k+1)}$ is the index set $I^{(k)}$ with the subset 
$\{r-k-1,r+k+1\}$ removed, and note that
\begin{align*}
x_{r\pm j}  &\equiv d_{r \pm j}^{-1}[[y_{r \pm (j+1)},X_j^{(0)}],y_{r \pm (j+1)}]\\
y_{r\pm j} &\equiv [x_{r \pm (j+1)},[x_{r \pm j}, y_{r \pm (j+1)}]]\\
h_{r\pm j} &\equiv 2^{-1}[x_{r \pm j}, y_{r \pm j}]
\end{align*}
Finally, we show that the generators $x_{r+k+j}$, $y_{r+k+j}$ and
$h_{r+k+j}$ are in $\L$ for $j=2\dots N-r-k-1$, by iterating the
following set of recurrence relations for $j=2 \dots N-r-k-1$:
\begin{align*}
x_{r+k+j}  &=d_{r+k+j}^{-1}[h_{r+k+j-1},V_0^{(k+j-1)}]\\
y_{r+k+j}  &=[x_{r+k+j},h_{r+k+j-1}]\\
h_{r+k+j}  &= 2^{-1}[x_{r+k+j},y_{r+k+j}]\\
V_0^{(k+j)}&=V_0^{(k+j-1)}-d_{r+k+j}y_{r+k+j}
\end{align*}
Similarly, we show that the elements $x_{r-k-j}$, $y_{r-k-j}$ and 
$h_{r-k-j}$ are in $\L$ for $j = 2, \dots, r-k-1$, by setting 
$W_0^{(k-1)}=V_0^{(k+1)}$ and iterating the following recurrence 
relations for $j=2 \dots r-k-1$:
\begin{align*}
x_{r-k-j}  &= d_{r-k-j}^{-1}[h_{r-k-j+1},W_0^{(k-j+1)}]\\
y_{r-k-j}  &= [x_{r-k-j},h_{r-k-j+1}]\\
h_{r-k-j}  &= 2^{-1}[x_{r-k-j},y_{r-k-j}]\\
W_0^{(k-j)}&= W_0^{(k-j+1)}-d_{r-k-j}y_{r-k-j}
\end{align*}
We have now shown that $x_j$, $y_j \in \LL$ for $j=1 \dots N-1$ as
desired, completing the proof.
\end{proof}

We note that $d_{r-k-1}^2\neq d_{r+k+1}^2$ for some integer $k$ is always 
satisfied if the system dimension is odd, $N=2\ell+1$, no matter where
we place the actuator.  If $N=2\ell$ then $d_{r-k-1}^2=d_{r+k+1}^2$ for 
all $k$ is possible only if $r=\ell$, i.e., if the actuator is placed in 
the middle, and the coupling constants are symmetric around the centre, 
$d_{\ell-k}^2=d_{\ell+k}^2$ for all $k$.  For a spin chain with strictly 
isotropic Heisenberg interaction the requirement $\omega_r\neq 0$ is still
a problem if $d_{r-1}=d_{r+1}$ but controllability could be restored by 
engineering a local perturbation of the energy levels in the vicinity of 
the actuator, which may indeed achieved by the actuator itself. 

We can interpret these results in terms of transition graphs.  Given a 
system with a tridiagonal drift Hamiltonian $\op{H}_0$ with respect to 
some Hilbert space basis $\{\ket{n}:1,\ldots,N\}$, we can define a 
transition graph as before by taking the $N$ basis states as vertices 
and adding edges for each non-zero transition.  Since the Hamiltonian is 
tridiagonal the resulting graph is either a linear chain or disconnected.
Connectedness is a necessary condition, and the results above guarantee
controllability in the following cases.
\begin{itemize}
\item If the chain is connected and has odd length then the actuator can
      be placed anywhere provided the vertices associated with the
      controlled edge have different energy levels.
\item If the chain is connected and has even length then we must ensure 
      in addition that the system does not admit symplectic symmetry, 
      which is guaranteed as long as the actuator is not placed
      precisely in the middle of the chain.
\item For a uniform chain for which the energy levels in the interior
      of the chain are always degenerate, controllability is ensured by 
      placing the actuator near either end of the chain.
\end{itemize}
It is worth pointing out here that the transition frequencies of the
system need not be distinct.  In fact, the most of the energy levels 
can be degenerate as is usually the case for uniform spin chains.  We 
only require that the vertices of the controlled transition have
different energy levels, which can generally be achieved by placing 
the actuator near the end of chain.

\begin{figure}
\scalebox{0.4}{\includegraphics{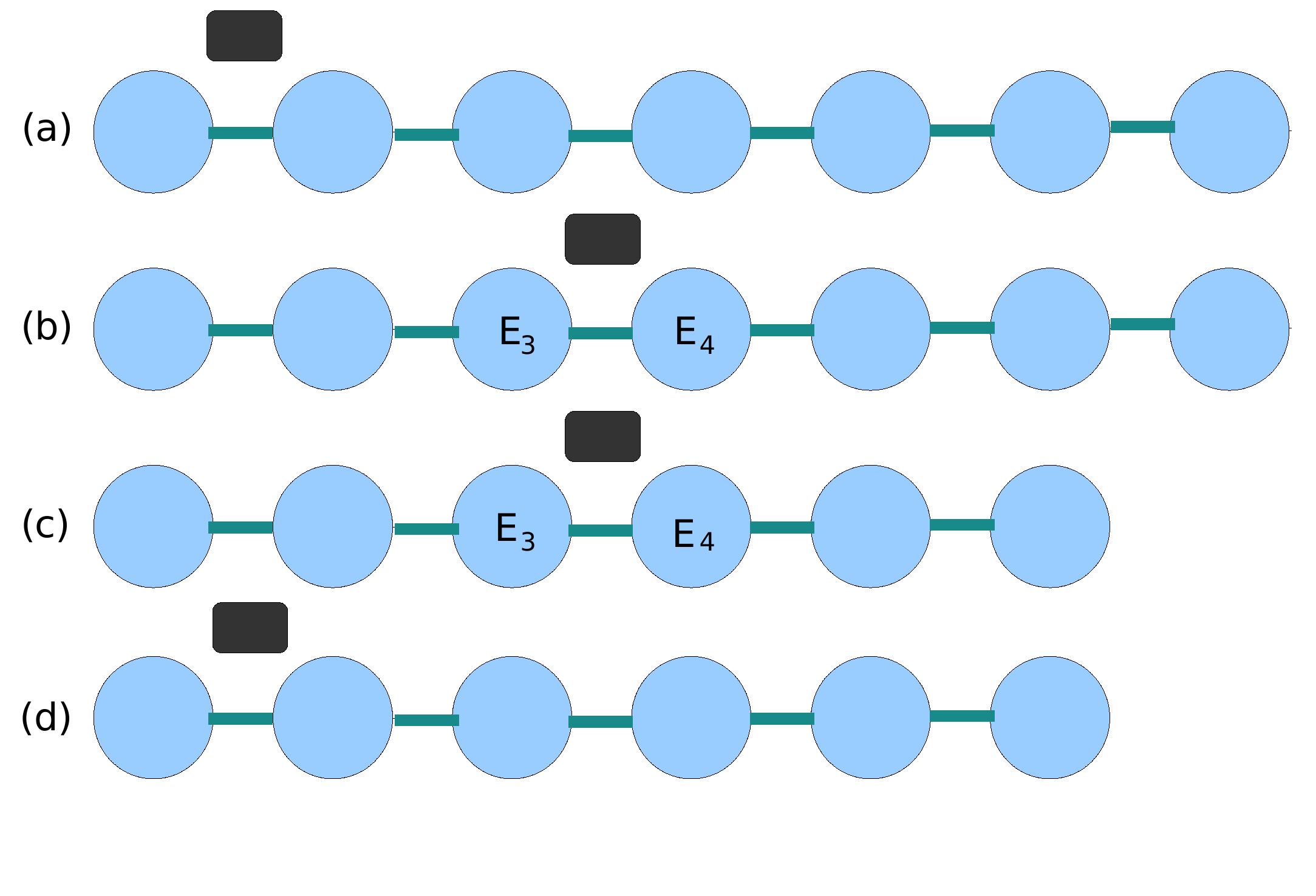}}
\caption{Controllability of a chain with single actuator: model systems
(a) and (d) are controllable, (b) is controllable if $E_3\neq E_4$, while
(c) is not controllable even if $E_3\neq E_4$ if chain has reflection 
symmetry, i.e., $|d_{2\ell-k}|=|d_k|$ and $E_{2\ell-k+1}=E_k$ for all 
$k$ as the associate dynamical Lie group has symplectic symmetry.}
\label{fig:2}
\end{figure}

\section{Constructive Control with Single Binary Switch Actuator}

The results of the previous section suggest that a single actuator is 
often sufficient to achive the same degree of controllability that is
achievable with many local actuators.  This result is not too surprising
on purely Lie algebraic grounds considering that two randomly choosen 
Hermitian $N \times N$ matrices, generically, will generate the entire 
Lie algebra $\uu(N)$.  Of course, the Hamiltonian matrices in our model 
system are far from random, and our Lie algebra calculations indeed show 
that for certain systems such as a spin chain with uniform, isotropic 
Heisenberg coupling between adjacent spins, controllability depends on 
the type of interaction and the placement of the actuator.  Nonetheless,
that a single local actuator in many cases results in the same degree of 
controllability than, say, $2N-1$ local actuators to individually control 
all of the energy levels and transitions, is rather surprising when one 
considers the substantially reduced control that such an actuator affords,
and it begs the question whether it is possible to control such a system 
constructively, i.e., whether we can find a local control field $f(t)$ 
that achieves the desired global system dynamics, and if there exists a
solution, whether it can be practically implemented.

The type of control functions that are feasible generally depends on the
specifics of the system and actuator.  For laser-controlled quantum
dots, for example, the availability of pulse shaping technology and the
demonstrated superiority of shaped pulses over simple pulses in certain
settings, suggests optimization routines designed to find an optimal
time-dependent pulse shape $f(t)$, and many such algorithms based on
gradients and variational techniques have been proposed (see for
example~\cite{JCP120p5509,PRA61n012101,JMR172p296}).  For many other
systems, especially voltage gate controlled systems, however, it is
generally difficult to implement complicated time-varying potentials,
and simple, piecewise constant controls that can be approximated by
square pulses are preferable.  In the following we consider the simplest
type of such an actuator, a binary switch that switches the voltage
between two possible values, corresponding to two fixed Hamiltonians
\begin{subequations}
\begin{align}
  H^{(1)} &= A_0 + f_0 A_1, \\
  H^{(2)} &= A_0 + f_1 A_1. 
\end{align}
\end{subequations}
Given a sequence of switching times $\vec{t}=(t_1,\ldots,t_K)$ the 
corresponding evolution of the system is given by
\begin{equation}
 U(\vec{t}) = U^{(1)}(t_1) U^{(2)}(t_2) \ldots U^{(1)}(t_{K-1}) U^{(2)}(t_K)
\end{equation}
where $U^{(m)}(t_k) = \exp(-i t_k H^{(m)})$ for $m=1,2$.  The control
task in this case is reduced to find the switching times $\vec{t}$ to
accomplish a desired task.  Although analytical expressions for the
optimal switching times are generally very difficult to obtain for all
but very simple systems, numerical optimization techniques can be used 
to find suitable controls, and we have found them to be surprisingly 
effective in many cases.

As a specific example, we consider the first excitation subspace of a
spin chain of length four with a single binary switch actuator placed 
between spins one and two, i.e., $r=1$.  This system is controllable
with a single actuator at $r=1$ according to Theorem~\ref{thm:1}.  To
show that we can constructively control this system with a single binary 
switch actuator, we find switching time sequences for a complete set of 
generators of $\SU(4)$.  Interpreting the first excitation subspace 
of the chain as a two-qubit system by setting 
\begin{equation*}
 \ket{0}=\ket{00}, \; \ket{1}=\ket{01}, \; \ket{2}=\ket{10}, \; \ket{3}=\ket{11},
\end{equation*} 
we show that it is possible to find vectors $\vec{t}^{(s)}$ such that
\begin{equation}
 \norm{U_T^{(s)} - U(\vec{t}^{(s)})} \le 10^{-4}
\end{equation}
for the following set of six target operators
\begin{equation}
 \label{eq:targ}
  U_T^{(s)} \in \{I\otimes I, \; \Had \otimes I, \; T\otimes I, \; 
  I\otimes \Had, \; I\otimes T, \; \CNOT\},
\end{equation}
where $I$ is identity operator on a single two-level subspace (qubit),
$T=\exp(-i\pi/8\sz)$ is a $\pi/8$ phase gate, and $\Had$ and $\CNOT$ are
the Hadamard and $\CNOT$ gate, respectively,
\begin{equation*}
 \Had = \frac{1}{\sqrt{2}}\begin{pmatrix}1 &-1\\ 1 & 1 \end{pmatrix},
 \quad
 \CNOT= e^{-i\pi/4}\begin{pmatrix} I_2 & 0 \\ 0 & \sx \end{pmatrix}, 
\end{equation*} 
with $\sx$ and $\sz$ being the usual Pauli matrices
\begin{equation*}
  \sx = \begin{pmatrix} 0 & 1 \\ 1 & 0  \end{pmatrix}, \quad
  \sz = \begin{pmatrix} 1 & 0 \\ 0 & -1 \end{pmatrix}.
\end{equation*}
The set of target operators~(\ref{eq:targ}) was chosen because it is a 
universal set of elementary gates in that any other $\SU(4)$ operator can
be constructed from these elementary gates, and the ability to implement 
a universal set of gates for $\SU(4)$ is equivalent to (density operator)
controllability of the system.

Table~\ref{table:1} shows the time vectors $\vec{t}^{(s)}$, as well as the 
gate operation times $T=\sum_{k=1}^K t_k$ and gate errors as defined above 
for a system with uniform isotropic Heisenberg coupling.  The optimization 
was performed using a Nelder-Mead downhill simplex algorithm~\cite{simplex} 
with multiple initial simplices.  The table shows that it is possible to 
implement all of the six elementary gates with a fidelity $\ge 99.99$\% with 
no more than $20$ switches of a single on-off switch actuator in approximately 
$40$ time units each, a surprisingly good result considering the minimal
nature of the available control.  Since  solutions are obviously not unique,
and the minimum control time or number of switches required to achieve the 
control objective are unknown, even better solutions probably exist.  The 
non-uniqueness of the solutions can be exploited to satisfy additional 
constraints such as minimum pulse lengths (switching cannot be arbitrarily
fast) etc.

\begin{table*}
\begin{center}
\begin{tabular}{lcccccc}
  & $I\otimes I$ 
  & $\mbox{Had} \otimes I$ 
  & $T\otimes I$ 
  & $I\otimes \mbox{Had}$  
  & $I\otimes T$ 
  &  \mbox{CNOT} \\\hline
error & 6.02407e-05 & 9.93462e-06 & 9.41944e-10 & 8.86637e-07 & 1.10773e-06 & 1.31773e-06\\
duration & 40.5351 & 37.9537 & 41.166 & 41.1328 & 42.5368 & 39.3569\\\hline
$t_1$ & 0.731996 & 3.94518 & 1.79446 & 3.08601 & 1.30764 & 3.34576\\
$t_2$ & 2.03884 & 2.20021 & 1.79932 & 3.13305 & 1.1069 & 0.0179813\\
$t_3$ & 3.52727 & 0.0384191 & 0.0730935 & 0.701478 & 0.518925 & 2.59171\\
$t_4$ & 1.38628 & 1.07432e-07 & 1.71885 & 3.62498 & 5.85085 & 3.23448\\
$t_5$ & 3.39919 & 0.680856 & 2.07051 & 2.45712 & 0.396729 & 1.46693\\
$t_6$ & 0.951534 & 3.04816 & 0.747468 & 0.68558 & 7.37392 & 0.212212\\
$t_7$ & 1.35113 & 1.292 & 1.84047 & 1.3746 & 1.13031 & 5.0851\\
$t_8$ & 0.575672 & 1.86256 & 2.53341 & 1.12801 & 1.16712 & 3.07975\\
$t_9$ & 3.38307 & 4.14879 & 4.73792 & 3.50997 & 0.802765 & 2.75667\\
$t_{10}$ & 0.0365974 & 0.356856 & 1.3432 & 1.92944 & 4.08279 & 0.439889\\
$t_{11}$ & 3.62131 & 1.02202 & 1.39084 & 5.57909 & 1.27132 & 3.25423\\
$t_{12}$ & 0.93505 & 0.0453206 & 0.320722 & 0.298252 & 2.70023 & 2.41685\\
$t_{13}$ & 1.75377 & 2.13701 & 4.15595 & 0.987279 & 4.67647 & 1.04768\\
$t_{14}$ & 5.19515 & 1.24291 & 0.533115 & 0.26934 & 0.705919 & 1.31426\\
$t_{15}$ & 4.5099 & 0.101593 & 1.03574 & 1.7998 & 1.01477 & 2.6859\\
$t_{16}$ & 1.01899 & 4.40131 & 7.58673 & 4.66334 & 1.78438 & 0.732592\\
$t_{17}$ & 4.01314 & 1.07241 & 4.77061 & 0.135612 & 1.02283 & 0.16703\\
$t_{18}$ & 0.991019 & 5.83516 & 0.857316 & 1.31499 & 1.02426 & 0.770284\\
$t_{19}$ & 0.705316 & 1.6229 & 1.7735 & 3.38444 & 2.16687 & 2.32287\\
$t_{20}$ & 0.409887 & 2.89999 & 0.082803 & 1.07044 & 2.43177 & 2.41471\\
\end{tabular}
\end{center}
\caption{Gate errors ($1-$gate fidelity), total time $T$ required to
 implement respective gates, and vector of switching times $t_k$ to 
 implement a universal set of elementary gates with 20 switches for
(the first excitation subspace of) a uniform isotropic Heisenberg spin 
chain of length four.}
\label{table:1}
\end{table*}

\section{Conclusion}

We have shown that a certain class of systems of Hilbert space dimension
$N$ is controllable with a single local actuator.  In particular, the results show that it is usually 
not necessary to be able to control all transitions, and a single local 
actuator in fact suffices in most cases to achieve controllability. That is to say, we do not require $N-1$ or
more independent local actuators, or global actuators acting on the
entire system.  

The results establish theoretical minimum requirements 
for controllability for a class of systems that includes many types of 
spin chains and other systems with non-trivial fixed couplings between 
adjacent elements.  As systems with fixed interactions are generally much easier to engineer 
than systems with individually tunable transitions, this is a promising 
result.

Although the controllability proof is an existence proof, we have further 
demonstrated that is is possible to constructively control a system with 
a local actuator, even if the actuator is limited to a binary switch, for 
a four-level system, where we have shown that it is possible to implement 
a complete set of generators of $\SU(4)$ with fidelities $\ge 99.99$\% 
using a single binary switch actuator, with no more than 20 switches per 
gate required.  Although it would be desirable to have analytic expressions 
for the switching times, it appears that numerical optimization techniques 
are quite effective in finding suitable controls.

Numerical simulations extending the technique to systems with non-tridiagonal Hamiltonians suggest that constructive 
control is generally still possible, and similar strong controllability 
results can almost certainly be obtained using very similar arguments for 
these systems.  This class would include for example interesting systems such as spin-chains with non-nearest neighbour couplings. Beyond the extension of generic controllability results 
to other classes of Hamiltonians, interesting questions for future work---%
the answers to which could point the way to achieving effective control
with much simpler control system designs---include what type of systems 
can be effectively controlled with a single local actuator, whether the 
placement of the actuator matters, and how different forms of coupling between the system and the actuator affect the control outcomes in practice.

\acknowledgments

SGS acknowledges support from an EPSRC Advanced Research Fellowship,
the \mbox{EPSRC QIP IRC} and Hitachi, and is currently also a Marie
Curie Fellow under the European Union Knowledge Transfer Programme 
MTDK-CT-2004-509223  PJP is supported by an \mbox{EPSRC} Project 
Studentship.

\bibliography{/home/sonia/archive/bibliography/References}
\bibliographystyle{prsty}
\end{document}